\documentclass[%
 reprint,
 amsmath,amssymb,
 aps,
]{revtex4-1}
\usepackage{amsmath}
\usepackage{booktabs}
\usepackage{array}
\usepackage{amsthm}
\usepackage{float}
\usepackage{graphicx}
\usepackage{dcolumn}
\usepackage{bm}

\begin{document}

\preprint{APS/123-QED}

\title{Nonlocality without entanglement in general multipartite quantum systems}

\author{Xiao-Fan Zhen$^{1}$}
\author{Shao-Ming Fei$^{4}$}
\author{Hui-Juan Zuo$^{1,2,3}$}
\email{huijuanzuo@163.com}

\affiliation{
$^{1}$School of Mathematical Sciences, Hebei Normal University, Shijiazhuang, 050024, China\\
$^{2}$Hebei Key Laboratory of Computational Mathematics and Applications, Shijiazhuang, 050024, China\\
$^{3}$Hebei International Joint Research Center for Mathematics and Interdisciplinary Science, Shijiazhuang, 050024, China\\
$^{4}$School of Mathematical Sciences, Capital Normal University, Beijing, 100048, China
}%


\begin{abstract}
The construction of nonlocal sets of quantum states has attracted much attention in recent years. We first introduce two lemmas related to the triviality of orthogonality-preserving local measurements. Then we propose a general construction of nonlocal set of $n(d-1)+1$ orthogonal product states in $(\mathbb{C}^{d})^{\otimes n}$. The sets of nonlocal orthogonal product states are also put forward for multipartite quantum systems with arbitrary dimensions. Our construction gives rise to nonlocal sets of orthogonal product states with much less members and thus reveals the phenomenon of nonlocality without entanglement more efficiently.

\begin{description}
\item[PACS numbers]
03.65.Ud, 03.67.Mn
\end{description}
\end{abstract}

\pacs{Valid PACS appear here}
\maketitle


\section{\label{sec:level1}Introduction\protect}

In recent years, local discrimination of quantum states has attracted more and more attention. A set of orthogonal multipartite quantum states is said to be locally indistinguishable if they cannot be discriminated through local operations and classical communication (LOCC), which reflects an important feature of quantum mechanics called nonlocality. Fruitful results have been presented to locally indistinguishable orthogonal product states and orthogonal entangled states \cite{Bennett1999,Bennett 1999UPB,Walgate2000,Ghosh2001,Walgate2002,Alon2001,DiVincenzo 2003CPM,Ghosh2004,Chen2004,Fan2004,Niset2006,Cohen2007,Xin2008,Yang2013,Zhangzc2014,Yu,Wangyl2015,Zhangzc2015,Zhangxq2016,Xu2016,
Zhangzc2016,Wangyl2017,Zhangzc2017,Halder2018,Jiang2020,Xu2021,Zuo2022} with wide applications in quantum cryptographic protocols \cite{Zhangzc2018,Li2019QIP,Li2019PRA,Yang2015,Wang2017,Jiangdh2020}.
An unextendible product basis (UPB) is given by a set of orthogonal product vectors whose complementary subspace contains no product basic states. The UPBs reflect the well-known phenomenon of nonlocality without entanglement and have an interesting feature, namely the states in the UPB cannot be perfectly distinguished by local positive operator-valued
measures(POVMs) and classical communication \cite{Bennett 1999UPB, DiVincenzo 2003CPM}. The concept of strong nonlocality was further introduced by Halder {\it et al}. \cite{Halder2019} based on locally irreducible quantum states. In multipartite quantum systems, there exist a number of quantum states that manifest strong quantum nonlocality \cite{Halder2019,Zhangzc2019,Rout2019,Yuan2020,Shi2020}.

Especially, since Bennett {\it et al}. put forward the phenomenon of nonlocality without entanglement in Ref. \cite{Bennett1999}, the construction of nonlocal state sets with a smaller number has become a major task to illustrate the nonlocality. An important method was provided to verify the local indistinguishability of orthogonal product states in Ref. \cite{Walgate2002}, showing that no matter which local party goes first to be operated, only trivial measurements can be performed. Subsequently, a series of progress had been made toward to the research on bipartite quantum systems. Zhang {\it et al}. \cite{Zhangzc2016} found that $2n-1$ orthogonal product states cannot be locally distinguished in $\mathbb{C}^{m}\otimes \mathbb{C}^{n}$, where $3\leq m\leq n$.

It is more difficult to construct LOCC indistinguishable states in multipartite quantum systems.
In Ref. \cite{Bennett 1999UPB} the authors presented a simple lower bound on the size $n$ of UPB in $\mathbb{C}^{d_{1}}\otimes \mathbb{C}^{d_{2}}\otimes \cdots\otimes \mathbb{C}^{d_{n}}$, $n\geq \sum_{i}(d_i-1)+1$. Further, Alon {\it et al}. investigated all the cases of $n=\sum_{i}(d_i-1)+1$ and showed that the minimum possible cardinality of a UPB is $\sum_{i}(d_i-1)+1$ when $d_i$ is odd \cite{Alon2001}. In Ref. \cite{Niset2006}, Niset {\it et al}. provided the nonlocal orthogonal product bases in multipartite quantum systems with arbitrary dimensions. Then, Xu {\it et al}. proposed a general method to construct a complete $n$-partite product basis with only $2n$ members in Ref. \cite{Xu2016}. By using three-dimensional cubes \cite{Wangyl2017}, Wang {\it et al.} constructed a number of $2(n_{1}+n_{3})-3$ orthogonal product states in $\mathbb{C}^{n_{1}}\otimes \mathbb{C}^{n_{2}}\otimes \mathbb{C}^{n_{3}}$, which cannot be perfectly distinguished by LOCC. Then, Zhang {\it et al}. showed the phenomenon of multipartite quantum nonlocality without entanglement in Ref. \cite{Zhangzc2017}. Halder {\it et al}. \cite{Halder2018} put forward an alternative construction of $2n(d+1)$ locally indistinguishable product states in multipartite quantum systems.

Recently, Jiang {\it et al}. \cite{Jiang2020} successfully proposed a less number of $\sum_{i=1}^{n}(2d_{i}-3)+1$ locally indistinguishable quantum states in $\mathbb{C}^{d_{1}}\otimes \mathbb{C}^{d_{2}}\otimes \cdots\otimes \mathbb{C}^{d_{n}}$. For general multipartite quantum systems, it is still a challenging problem to find locally indistinguishable quantum state sets with less number of orthogonal product states in which any subsystem can only perform a trivial POVM.

In this manuscript, we put forward the construction of LOCC indistinguishable orthogonal product states in multipartite quantum systems. Based on the stopper state, we provide two important lemmas to show that an orthogonality-preserving local measurement must be trivial. First, we construct locally indistinguishable orthogonal product states with less number in the equal-dimensional multipartite quantum systems. Then we propose the set of $\sum\limits_{i=2}^{n-1}d_{i}+2d_{n}-n+1$ nonlocal multipartite orthogonal product states in $\mathbb{C}^{d_{1}}\otimes \mathbb{C}^{d_{2}}\otimes \cdots\otimes \mathbb{C}^{d_{n}}$ for $3\leq d_{1}\leq d_{2}\cdots\leq d_{n}$ and $n\geq 3$. Compared with the previous results, our results demonstrate a better cognition of quantum nonlocality without entanglement.

\theoremstyle{remark}
\newtheorem{definition}{\indent Definition}
\newtheorem{lemma}{\indent Lemma}
\newtheorem{theorem}{\indent Theorem}
\newtheorem{proposition}{\intent Proposition}
\newtheorem{corollary}{\indent Corollary}
\def\QEDclosed{\mbox{\rule[0pt]{1.3ex}{1.3ex}}}|
\def\QED{\QEDclosed}
\def\proof{\indent{\em Proof}.}
\def\endproof{\hspace*{\fill}~\QED\par\endtrivlist\unskip}

\section{Preliminaries}

We take the computational basis $\{|i\rangle\}_{i=0}^{d_{k}-1}$ for each $d_{k}$-dimensional subsystem. For simplicity, we denote the state $\frac{1}{\sqrt{n}}(|i_1\rangle \pm |i_2\rangle\pm \cdots \pm |i_n\rangle)$ as $|i_1\pm i_2\pm\cdots \pm i_n\rangle$ and $\mathbb Z_{d_{k}}=\{0, 1,\cdots, d_{k}-1\}$.
Let
\begin{equation}\label{equa1}
\begin{aligned}
|\phi_{p_{i}^{k},i}\rangle= \otimes_{k=1}^{n}(\sum_{i_{k}\in \mathbb Z_{d_{k}}} p_{i}^{k}|i_{k}\rangle)
\end{aligned}
\end{equation}
be $n$-partite orthogonal product states, where $p_{i}^{k}=0$, $-1$ or $1$ and at least one $p_{i}^{k}\neq 0$. In particular, the stopper state $|S\rangle$ is given by $p_{i}^{k}=1$,
\begin{equation}
\begin{aligned}
|S\rangle &= \otimes_{k=1}^{n}(\sum_{i_{k}\in \mathbb Z_{d_{k}}}|i_{k}\rangle).
\end{aligned}
\end{equation}

To certify the nonlocality of a set of multipartite orthogonal product states, one performs local POVM on these states such that the post-measurement states remain orthogonal. Each POVM element  $M_{k}^{\dagger}M_{k}$ can be expressed as a $d_{k}\times d_{k}$ matrix $E_k$ in the computational basis. A measurement is trivial if all the POVM elements are proportional to the identity operator. The set of orthogonal quantum states cannot be distinguished by LOCC if the POVMs are all trivial.

We first bring forward two lemmas to obtain the zero and diagonal entries of the matrix $E_t=(m_{i,j}^{t})_{i,j\in \mathbb Z_{d_{t}}}$. Consider a pair of multipartite orthogonal product states $\{|\phi_{p_{i}^{k},i}\rangle,~|\phi_{p_{j}^{k},j}\rangle\}$ whose $n-1$ subsystems are not mutually orthogonal except for the $t$-th subsystem. A local POVM is preformed on the $t$-th subsystem and identity operators are preformed on the rest $n-1$ subsystems.

\begin{lemma}({\bf Zero entries})\label{lem:zero}
If there exist only one $p_{a}^{t}\neq 0$ for $|\phi_{p_{i}^{k},i}\rangle$ and one $p_{b}^{t}\neq 0$ for $|\phi_{p_{j}^{k},j}\rangle$, then
\begin{equation}
\begin{aligned}
m_{a,b}^{t}= 0,
\end{aligned}
\end{equation}
where $0\leq a\neq b\leq d_{t}-1$, $1\leq t \leq n$.
\end{lemma}

\begin{proof} Since the post-measurement states $\{\mathbb {I}_{1}\otimes \cdots \otimes M_{t}\otimes \cdots \otimes \mathbb {I}_{n}|\phi_{p_{i}^{k},i}\rangle\}$ and $\{\mathbb {I}_{1}\otimes \cdots \otimes M_{t}\otimes \cdots \otimes \mathbb {I}_{n}|\phi_{p_{j}^{k},j}\rangle\}$ should be mutually orthogonal, we have
\begin{equation}
\begin{aligned}
\langle \phi_{p_{i}^{k},i}|(\mathbb {I}_{1}\otimes \cdots \otimes M_{t}^{\dagger}M_{t}\otimes \cdots \otimes \mathbb {I}_{n})|\phi_{p_{j}^{k},j}\rangle= 0,
\end{aligned}
\end{equation}
namely,
\begin{equation}
\begin{aligned}
(\sum_{i=0}^{d_{t}-1}p_{i}^{t}\langle i_{t}|)E_{t}(\sum_{j=0}^{d_{t}-1}p_{j}^{t}|j_{t}\rangle)&= 0.
\end{aligned}
\end{equation}
Because only one $p_{a}^{t}\neq 0$ for $|\phi_{p_{i}^{k},i}\rangle$ and one $p_{b}^{t}\neq 0$ for $|\phi_{p_{j}^{k},j}\rangle$, we have
$\langle a_{t}|E_{t}|b_{t}\rangle=0$. Hence, $m_{a,b}^{t}=0$ for any $0\leq a\neq b\leq d_{t}-1$ and $1\leq t \leq n$.
\end{proof}

\begin{lemma}({\bf Diagonal entries})\label{lem:Dia}
For $|\phi_{p_{j}^{k},j}\rangle=|S\rangle$, if all the off-diagonal entries of the matrix $E_{t}$ are zeros, and there exist only one $p_{a}^{t}=1$ and one $p_{b}^{t}=-1$ for $|\phi_{p_{i}^{k},i}\rangle$, then
\begin{equation}
\begin{aligned}
m_{a,a}^{t}= m_{b,b}^{t},
\end{aligned}
\end{equation}
where $0\leq a\neq b\leq d_{t}-1$, $1\leq t \leq n$.
\end{lemma}

\begin{proof} Because the post-measurement states $\{\mathbb{I}_{1}\otimes \cdots \otimes M_{t}\otimes \cdots \otimes \mathbb{I}_{n}|\phi_{p_{i}^{k},i}\rangle\}$ and $\{\mathbb{I}_{1}\otimes \cdots \otimes M_{t}\otimes \cdots \otimes \mathbb{I}_{n}|S\rangle\}$ are mutually orthogonal, we have
\begin{equation}
\begin{aligned}
\langle \phi_{p_{i}^{k},i}|(\mathbb {I}_{1}\otimes \cdots \otimes M_{t}^{\dagger}M_{t}\otimes \cdots \otimes \mathbb {I}_{n})|S\rangle= 0,
\end{aligned}
\end{equation}
which gives rise to
\begin{equation}
\begin{aligned}
(\sum_{i=0}^{d_{t}-1}p_{i}^{t}\langle i_{t}|)E_{t}(\sum_{j=0}^{d_{t}-1}|j_{t}\rangle)&= 0.
\end{aligned}
\end{equation}
Since only one $p_{a}^{t}= 1$ and one $p_{b}^{t}= -1$ for $|\phi_{p_{i}^{k},i}\rangle$ and all $m_{i,j}^{t}= 0$ for $0\leq i\neq j\leq d_{t}-1$, we obtain
\begin{equation}
\begin{aligned}
(\langle a_{t}|-\langle b_{t}|)E_{t}(|a_{t}\rangle+|b_{t}\rangle)&= 0.
\end{aligned}
\end{equation}
Hence, $m_{a,a}^{t}= m_{b,b}^{t}$ for any $0\leq a\neq b\leq d_{t}-1$ and  $1\leq t \leq n$.
\end{proof}

\section{Construction Of Nonlocal States in $(\mathbb{C}^{d})^{\otimes n}$}

In this section, we propose a nonlocal set of orthogonal product states in $(\mathbb{C}^{d})^{\otimes n}$. First we illustrate the construction for a four-partite case.

\begin{lemma}\label{lem:$4(d-1)+1$} The set of following $4(d-1)+1$ orthogonal product states is indistinguishable by LOCC in $\mathbb{C}^d\otimes \mathbb{C}^d\otimes \mathbb{C}^d\otimes \mathbb{C}^d$:
\begin{equation}
\begin{aligned}
|\phi_{i}\rangle&=|0-i\rangle_{1}|0\rangle_{2}|0\rangle_{3}|i\rangle_{4},\\
|\phi_{i+(d-1)}\rangle & =|i\rangle_{1}|0-i\rangle_{2}|0\rangle_{3}|0\rangle_{4}, \\
|\phi_{i+2(d-1)}\rangle & =|0\rangle_{1}|i\rangle_{2}|0-i\rangle_{3}|0\rangle_{4},\\
|\phi_{i+3(d-1)}\rangle & =|0\rangle_{1}|0\rangle_{2}|i\rangle_{3}|0-i\rangle_{4},\\
|\phi_{4(d-1)+1}\rangle & =|0+\cdots+(d-1)\rangle_{1}|0+\cdots+(d-1)\rangle_{2}\\& ~~~~|0+\cdots+(d-1)\rangle_{3}|0+\cdots+(d-1)\rangle_{4},
\end{aligned}
\end{equation}
where $i=1,2,\cdots, d-1$,~$d\geq 3$.
\end{lemma}

\begin{proof} Since the states are symmetric, we only need to prove that the measurement applied to the first subsystem is trivial.

Applying Lemma~\ref{lem:zero} to the states $|\phi_{i+(d-1)}\rangle$ and $|\phi_{i+2(d-1)}\rangle$, we obtain $m_{i,0}^{1}= m_{0,i}^{1}= 0$ for $1\leq i\leq d-1$ by the orthogonality. Similarly, for the states $|\phi_{i+(d-1)}\rangle$ and $|\phi_{j+(d-1)}\rangle$, we have $m_{i,j}^{1}= m_{j,i}^{1}= 0$, where $1\leq i\neq j\leq d-1$.

Consider the states $|\phi_{i}\rangle$ and $|\phi_{4(d-1)+1}\rangle$, we directly get $m_{0,0}^{1}= m_{i,i}^{1}$ with $1\leq i\leq d-1$ by Lemma~\ref{lem:Dia}.

In summary, $E_{1}=(m_{i,j}^{1})_{i,j\in \mathbb Z_{d}}$ is proportional to the identity matrix, so one cannot start with a nontrivial measurement on the first subsystem. Thus, the above $4(d-1)+1$ states cannot be distinguished by LOCC.
\end{proof}

In Ref. \cite{Jiang2020}, the authors presented a set of $8d-11$ indistinguishable orthogonal product states by LOCC in $\mathbb{C}^d\otimes \mathbb{C}^d\otimes \mathbb{C}^d\otimes \mathbb{C}^d$. Here we put forward a set of $4d-3$ indistinguishable orthogonal product states by LOCC in $\mathbb{C}^d\otimes \mathbb{C}^d\otimes \mathbb{C}^d\otimes \mathbb{C}^d$. Our construction has less number of states for $d>2$.

Based on the above structure, we propose the general construction in $(\mathbb{C}^{d})^{\otimes n}$.

\begin{theorem}\label{the:$n(d-1)+1$} The following $n(d-1)+1$ orthogonal product states in $(\mathbb{C}^{d})^{\otimes n}$ cannot be perfectly distinguished by LOCC:
\begin{equation}
\begin{aligned}
|\phi_{i}\rangle & =|0-i\rangle_{1}|0\rangle_{2}|0\rangle_{3}\cdots|0\rangle_{n-2}|0\rangle_{n-1}|i\rangle_{n},\\
|\phi_{i+(d-1)}\rangle & =|i\rangle_{1}|0-i\rangle_{2}|0\rangle_{3}\cdots|0\rangle_{n-2}|0\rangle_{n-1}|0\rangle_{n},\\
|\phi_{i+2(d-1)}\rangle & =|0\rangle_{1}|i\rangle_{2}|0-i\rangle_{3}\cdots|0\rangle_{n-2}|0\rangle_{n-1}|0\rangle_{n},\\
& ~~~~~~~~\cdots~ \cdots ~\cdots~~~~~~~~~~~~\\
|\phi_{i+(n-2)(d-1)}\rangle & =|0\rangle_{1}|0\rangle_{2}|0\rangle_{3}\cdots|i\rangle_{n-2}|0-i\rangle_{n-1}|0\rangle_{n},\\
|\phi_{i+(n-1)(d-1)}\rangle & =|0\rangle_{1}|0\rangle_{2}|0\rangle_{3}\cdots|0\rangle_{n-2}|i\rangle_{n-1}|0-i\rangle_{n},\\
|\phi_{n(d-1)+1}\rangle & =|0+\cdots+(d-1)\rangle_{1}|0+\cdots+(d-1)\rangle_{2}\\ &~~~~\cdots|0+\cdots+(d-1)\rangle_{n},
\end{aligned}
\end{equation}
where $i=1,\cdots,d-1$, $d\geq3,\, n\geq3$.
\end{theorem}

\begin{proof} First of all, we give all zero entries of the matrix $E_{t}$ according to Lemma~\ref{lem:zero}. Consider the $t$-th subsystem with $1\leq t\leq n-2$. From the states $|\phi_{i+t(d-1)}\rangle$ and $|\phi_{i+(t+1)(d-1)}\rangle$, we have $m_{i,0}^{t}= m_{0,i}^{t}= 0$ for $1\leq i\leq d-1$. From $|\phi_{i+t(d-1)}\rangle$ and $|\phi_{j+t(d-1)}\rangle$, we obtain $m_{i,j}^{t}= m_{j,i}^{t}= 0$, where $1\leq i \neq j\leq d-1$. Associated with the $(n-1)$-th subsystem, we get $m_{i,0}^{n-1}= m_{0,i}^{n-1}= 0$ with $1\leq i\leq d-1$ by using the states $|\phi_{i}\rangle$ and $|\phi_{i+(n-1)(d-1)}\rangle$. Considering the states $|\phi_{i+(n-1)(d-1)}\rangle$ and $|\phi_{j+(n-1)(d-1)}\rangle$, we have $m_{i,j}^{n-1}= m_{j,i}^{n-1}= 0$, where $1\leq i \neq j\leq d-1$.
As for the $n$-th subsystem, the states $|\phi_{i}\rangle$ and $|\phi_{i+(d-1)}\rangle$ provide the condition $m_{i,0}^{n}= m_{0,i}^{n}= 0$ with $1\leq i\leq d-1$. For the states $|\phi_{i}\rangle$ and $|\phi_{j}\rangle$, we have $m_{i,j}^{n}= m_{j,i}^{n}= 0$, where $1\leq i \neq j\leq d-1$.

By Lemma~\ref{lem:Dia}, we have the relations among the diagonal entries of $E_t$. For instance from the states $|\phi_{i+(t-1)(d-1)}\rangle$ and $|\phi_{n(d-1)+1}\rangle$, we directly deduce that $m_{0,0}^{t}= m_{i,i}^{t}$ with $1\leq i\leq d-1$ and $1\leq t\leq n$.

From the above analysis, all POVM elements $M_{t}^{\dagger}M_{t}$ are proportional to the identity matrix. One cannot start with a nontrivial measurement on any subsystem. This completes the proof.
\end{proof}

The authors in Ref. \cite{Jiang2020} presented a set of $n(2d-3)+1$ nonlocal product states in $(\mathbb{C}^{d})^{\otimes n}$. It is easily seen that the set of our $n(d-1)+1$ locally indistinguishable states has less number of states
for $d>2$.

\section{Construction of Nonlocal Sets in General Multipartite Systems}

In this section, we put forward a few nonlocal product sets in~$\mathbb{C}^{d_{1}}\otimes \mathbb{C}^{d_{2}}\otimes\cdots\otimes \mathbb{C}^{d_{n}}$. We first show our construction in arbitrary tripartite quantum systems.

\begin{lemma}\label{lea:d1d2d3} The following set of $d_{2}+2d_{3}-2$ orthogonal product states in $\mathbb{C}^{d_{1}}\otimes \mathbb{C}^{d_{2}}\otimes \mathbb{C}^{d_{3}}$ ($3\leq d_{1}\leq d_{2}\leq d_{3}$) is indistinguishable under LOCC:
\begin{equation}
\begin{aligned}
|\phi_{i}\rangle&=|0-i\rangle_{1}|0\rangle_{2}|i\rangle_{3},~~~1\leq i\leq d_{1}-1\\
|\phi_{i+(d_{1}-1)}\rangle & =|i\rangle_{1}|0-i\rangle_{2}|0\rangle_{3}, ~~~1\leq i\leq d_{1}-1\\
|\phi_{i+2(d_{1}-1)}\rangle & =|0\rangle_{1}|i\rangle_{2}|0-i\rangle_{3},~~~1\leq i\leq d_{2}-1\\
|\phi_{i+d_{1}+d_{2}-2}\rangle & =|1\rangle_{1}|0-i\rangle_{2}|i\rangle_{3},~~~d_{1}\leq i\leq d_{2}-1\\
|\phi_{i+d_{1}+d_{2}-2}\rangle & =|m\rangle_{1}|1\rangle_{2}|(i-1)-i\rangle_{3},~~~d_{2}\leq i\leq d_{3}-1\\
|\phi_{i+d_{2}+d_{3}-2}\rangle & =|0-2\rangle_{1}|0-2\rangle_{2}|i\rangle_{3},~~~d_{1}\leq i\leq d_{3}-1\\
|\phi_{d_{2}+2d_{3}-2}\rangle & =|0+1+\cdots+(d_{1}-1)\rangle_{1}\\ &~~~~|0+1+\cdots+(d_{2}-1)\rangle_{2}\\ &~~~~|0+1+\cdots+(d_{3}-1)\rangle_{3},
\end{aligned}
\end{equation}
where $m = 2$ ($m = 1$) when $i$ is even (odd).
\end{lemma}

\begin{proof} According to Lemma~\ref{lem:zero}, we first obtain that most off-diagonal elements of the matrix $E_{t}=(m_{i,j}^{t})_{i,j\in \mathbb Z_{d_{t}}}$ are zeros from Table \ref{1_result}:
\begin{table}[htbp]
    \newcommand{\tabincell}[2]{\begin{tabular}{@{}#1@{}}#2\end{tabular}}
    \centering
    \caption{\label{1_result}Zero entries of the matrix $E_{t}=(m_{i,j}^{t})_{i,j\in \mathbb Z_{d_{t}}}$.}
    \begin{tabular}{ccc}
        \toprule
        \hline
        \hline
        \specialrule{0em}{1.5pt}{1.5pt}
        ~~Pairs of states~~~&~~~~Zero entries~~~&~~~Range~~~\\
        \specialrule{0em}{1.5pt}{1.5pt}
        \midrule
        \hline
        \specialrule{0em}{1.5pt}{1.5pt}
        \tabincell{c}{$|\phi_{i+(d_{1}-1)}\rangle$ \\$|\phi_{i+2(d_{1}-1)}\rangle$}&\tabincell{c}{$m_{i,0}^{1}= m_{0,i}^{1}= 0$}&\tabincell{c}{$1\leq i\leq d_{1}-1$} \\
        \specialrule{0em}{3.5pt}{3.5pt}
        \tabincell{c}{$|\phi_{i+(d_{1}-1)}\rangle$ \\$|\phi_{j+(d_{1}-1)}\rangle$}&\tabincell{c}{$m_{i,j}^{1}= m_{j,i}^{1}= 0$}&\tabincell{c}{$1\leq i \neq j\leq d_{1}-1$} \\
        \specialrule{0em}{3.5pt}{3.5pt}
        \tabincell{c}{$|\phi_{i}\rangle$ \\$|\phi_{i+2(d_{1}-1)}\rangle$}&\tabincell{c}{$m_{i,0}^{2}= m_{0,i}^{2}= 0$}&\tabincell{c}{$1\leq i\leq d_{1}-1$} \\
        \specialrule{0em}{3.5pt}{3.5pt}
        \tabincell{c}{$|\phi_{i+2(d_{1}-1)}\rangle$ \\$|\phi_{j+2(d_{1}-1)}\rangle$}&\tabincell{c}{$m_{i,j}^{2}= m_{j,i}^{2}= 0$}&\tabincell{c}{$1\leq i \neq j\leq d_{2}-1$} \\
        \specialrule{0em}{3.5pt}{3.5pt}
        \tabincell{c}{$|\phi_{i}\rangle$ \\$|\phi_{i+(d_{1}-1)}\rangle$}&\tabincell{c}{$m_{i,0}^{3}= m_{0,i}^{3}= 0$}&\tabincell{c}{$1\leq i\leq d_{1}-1$} \\
        \specialrule{0em}{3.5pt}{3.5pt}
        \tabincell{c}{$|\phi_{d_{1}+1}\rangle$ \\$|\phi_{i+d_{2}+d_{3}-2}\rangle$}&\tabincell{c}{$m_{i,0}^{3}= m_{0,i}^{3}= 0$}&\tabincell{c}{$d_{1}\leq i \leq d_{3}-1$} \\
        \specialrule{0em}{3.5pt}{3.5pt}
        \tabincell{c}{$|\phi_{i}\rangle$ \\$|\phi_{j}\rangle$}&\tabincell{c}{$m_{i,j}^{3}= m_{j,i}^{3}= 0$}&\tabincell{c}{$1\leq i \neq j\leq d_{1}-1$} \\
        \specialrule{0em}{3.5pt}{3.5pt}
        \tabincell{c}{$|\phi_{i}\rangle$ \\$|\phi_{j+d_{2}+d_{3}-2}\rangle$}&\tabincell{c}{$m_{i,j}^{3}= m_{j,i}^{3}= 0$}&\tabincell{c}{$1\leq i \leq d_{1}-1$\\$d_{1}\leq j \leq d_{3}-1$} \\
        \specialrule{0em}{3.5pt}{3.5pt}
        \tabincell{c}{$|\phi_{i+d_{2}+d_{3}-2}\rangle$ \\$|\phi_{j+d_{2}+d_{3}-2}\rangle$}&\tabincell{c}{$m_{i,j}^{3}= m_{j,i}^{3}= 0$}&\tabincell{c}{$d_{1}\leq i\neq j\leq d_{3}-1$}\\
        \bottomrule
        \specialrule{0em}{1.5pt}{1.5pt}
        \hline
        \hline
    \end{tabular}
\end{table}

From the states $|\phi_{i+2(d_{1}-1)}\rangle$ and $|\phi_{j+d_{2}+d_{3}-2}\rangle$, we have  $\langle0|I_{1}|0-2\rangle \langle i|M_{2}^{\dagger}M_{2}|0-2\rangle \langle 0-i|I_{3}|j\rangle=0$ for $d_{1}\leq i=j\leq d_{2}-1$, that is, $\langle i|M_{2}^{\dagger}M_{2}|0-2\rangle=0$. Since $m_{i,2}^{2}= 0$, we get $m_{i,0}^{2}= m_{0,i}^{2}= 0$ with $d_{1}\leq i\leq d_{2}-1$.
\begin{table}[htbp]
    \newcommand{\tabincell}[2]{\begin{tabular}{@{}#1@{}}#2\end{tabular}}
    \centering
    \caption{\label{2_result}~Diagonal entries of $E_{t}=(m_{i,j}^{t})_{i,j\in \mathbb Z_{d_{t}}}$.}
    \begin{tabular}{ccc}
        \toprule
        \hline
        \hline
        \specialrule{0em}{1.5pt}{1.5pt}
        ~~~Pairs of states~~~&~~~Diagonal entries~~&~~Range~~~\\
        \specialrule{0em}{1.5pt}{1.5pt}
        \midrule
        \hline
        \specialrule{0em}{1.5pt}{1.5pt}
        \tabincell{c}{$|\phi_{i}\rangle$ \\$|\phi_{d_{2}+2d_{3}-2}\rangle$}&\tabincell{c}{$m_{0,0}^{1}= m_{i,i}^{1}$}&\tabincell{c}{$1\leq i\leq d_{1}-1$} \\
        \specialrule{0em}{3.5pt}{3.5pt}
        \tabincell{c}{$|\phi_{i+(d_{1}-1)}\rangle$ \\$|\phi_{d_{2}+2d_{3}-2}\rangle$}&\tabincell{c}{$m_{0,0}^{2}= m_{i,i}^{2}$}&\tabincell{c}{$1\leq i\leq d_{1}-1$} \\
        \specialrule{0em}{3.5pt}{3.5pt}
        \tabincell{c}{$|\phi_{i+d_{1}+d_{2}-2}\rangle$ \\$|\phi_{d_{2}+2d_{3}-2}\rangle$}&\tabincell{c}{$m_{0,0}^{2}= m_{i,i}^{2}$}&\tabincell{c}{$d_{1}\leq i\leq d_{2}-1$} \\
        \specialrule{0em}{3.5pt}{3.5pt}
        \tabincell{c}{$|\phi_{i+2(d_{1}-1)}\rangle$ \\$|\phi_{d_{2}+2d_{3}-2}\rangle$}&\tabincell{c}{$m_{0,0}^{3}= m_{i,i}^{3}$}&\tabincell{c}{$1\leq i \leq d_{2}-1$} \\
        \specialrule{0em}{3.5pt}{3.5pt}
        \tabincell{c}{$|\phi_{i+d_{1}+d_{2}-2}\rangle$ \\$|\phi_{d_{2}+2d_{3}-2}\rangle$}&\tabincell{c}{$m_{(i-1),(i-1)}^{3}= m_{i,i}^{3}$}&\tabincell{c}{$d_{2}\leq i\leq d_{3}-1$}\\
        \bottomrule
        \specialrule{0em}{1.5pt}{1.5pt}
        \hline
        \hline
    \end{tabular}
\end{table}

By Lemma~\ref{lem:Dia}, all the diagonal entries of the matrix $E_{t}$ are equal from Table \ref{2_result}. Thus, nobody can start with a non-trivial measurement on any subsystem. This completes the proof.
\end{proof}

In recent years, there has been a lot of research on nonlocal orthogonal product states in $\mathbb{C}^{d_{1}}\otimes \mathbb{C}^{d_{2}}\otimes \mathbb{C}^{d_{3}}$. The previous set of $2(d_{1}+d_{3})-3$ nonlocal states was constructed by Wang {\it et al}. \cite{Wangyl2017}. Our construction is clearly superior to this result for $d_{2}\leq 2d_{1}-1$. Moreover, we have the following general conclusions.

\begin{theorem}
\label{the:dn} The following set of $\sum\limits_{i=2}^{n-1}d_{i}+2d_{n}-n+1$ orthogonal product states in $\mathbb{C}^{d_{1}}\otimes\mathbb{C}^{d_{2}}\otimes\cdots\otimes\mathbb{C}^{d_{n}}$ cannot be perfectly distinguished by LOCC for $3\leq d_{1}\leq d_{2}\cdots\leq d_{n}$ and $n\geq 3$:
\begin{widetext}
\begin{equation}
\begin{aligned}
\mathcal{B}_{1}:~\{|\phi_{i}\rangle&=|0-i\rangle_{1}
|0\rangle_{2}|0\rangle_{3}\cdots|0\rangle_{n-2}|0\rangle_{n-1}|i\rangle_{n}\mid~i\in[1, d_{1}-1]\},\\
\mathcal{B}_2:~\{|\phi_{i+(d_{1}-1)}\rangle& =|i\rangle_{1}|0-i\rangle_{2}|0\rangle_{3}\cdots
|0\rangle_{n-2}|0\rangle_{n-1}|0\rangle_{n}\mid ~i\in[1, d_{1}-1]\},\\
\mathcal{B}_3:~\{|\phi_{i+2(d_{1}-1)}\rangle& =|0\rangle_{1}|i\rangle_{2}|0-i\rangle_{3}\cdots
|0\rangle_{n-2}|0\rangle_{n-1}|0\rangle_{n}\mid~i\in[1, d_{2}-1]\},\\
& ~~~~~~~~\cdots~ \cdots ~\cdots~~~~~~~~~~~~\\
\mathcal{B}_{n-1}:~\{|\phi_{i+(d_{1}-1)+\sum_{i=1}^{n-3}(d_{i}-1)}\rangle& =|0\rangle_{1}|0\rangle_{2}|0\rangle_{3}\cdots
|i\rangle_{n-2}|0-i\rangle_{n-1}|0\rangle_{n}\mid~i\in[1, d_{n-2}-1]\},\\
\mathcal{B}_{n}:~\{|\phi_{i+(d_{1}-1)+\sum_{i=1}^{n-2}(d_{i}-1)}\rangle& =|0\rangle_{1}|0\rangle_{2}|0\rangle_{3}\cdots
|0\rangle_{n-2}|i\rangle_{n-1}|0-i\rangle_{n}\mid~i\in[1, d_{n-1}-1]\},\\
\mathcal{B}_{n+1}:~\{|\phi_{i+\sum_{i=1}^{n-1}(d_{i}-1)}\rangle&
=|1\rangle_{1}|0-i\rangle_{2}|i\rangle_{3}|0\rangle_{4}\cdots
|0\rangle_{n-2}|0\rangle_{n-1}|0\rangle_{n}\mid~i\in[d_{1}, d_{2}-1]\},\\
\mathcal{B}_{n+2}:~\{|\phi_{i+\sum_{i=1}^{n-1}(d_{i}-1)}\rangle& =|0\rangle_{1}|1\rangle_{2}|0-i\rangle_{3}|i\rangle_{4}\cdots
|0\rangle_{n-2}|0\rangle_{n-1}|0\rangle_{n}\mid~i\in[d_{2}, d_{3}-1]\},\\
& ~~~~~~~~\cdots~ \cdots ~\cdots~~~~~~~~~~~~\\
\mathcal{B}_{2n-2}:~\{|\phi_{i+\sum_{i=1}^{n-1}(d_{i}-1)}\rangle& =|0\rangle_{1}|0\rangle_{2}|0\rangle_{3}|0\rangle_{4}\cdots
|1\rangle_{n-2}|0-i\rangle_{n-1}|i\rangle_{n}\mid~i\in[d_{n-2}, d_{n-1}-1]\},\\
\mathcal{B}_{2n-1}:~\{|\phi_{i+\sum_{i=1}^{n-1}(d_{i}-1)}\rangle& =|m\rangle_{1}|0\rangle_{2}|0\rangle_{3}\cdots|0\rangle_{n-2}
|1\rangle_{n-1}|(i-1)-i\rangle_{n}\mid~i\in[d_{n-1}, d_{n}-1]\},\\
\mathcal{B}_{2n}:~\{|\phi_{i+\sum_{i=2}^{n}(d_{i}-1)}\rangle& =|0-2\rangle_{1}|0\rangle_{2}|0\rangle_{3}\cdots|0\rangle_{n-2}
|0-2\rangle_{n-1}|i\rangle_{n}\mid~i\in[d_{1}, d_{n}-1]\},\\
\mathcal{B}_{2n+1}:~\{|\phi_{\sum_{i=2}^{n-1}d_{i}+2d_{n}-n+1}\rangle& =|0+1+\cdots+(d_{1}-1)\rangle_{1}|0+1+\cdots+(d_{2}-1)\rangle_{2}\\ &~~~~\cdots|0+1+\cdots+(d_{n-1}-1)\rangle_{n-1}|0+1+\cdots+(d_{n}-1)\rangle_{n}\},
\end{aligned}
\end{equation}
where $m = 2$ ($m = 1$) when $i$ is even (odd).
\end{widetext}
\end{theorem}

\begin{proof} By Lemma~\ref{lem:zero}, we have the most zero entries of the matrix $E_{t}= M_{t}^{\dagger}M_{t}$ from Table \ref{3_result}.

According to the two groups of states $\mathcal {B}_{n}$ and $\mathcal {B}_{2n}$, we have $\langle0|I_{1}|0-2\rangle\langle0|I_{2}|0\rangle$ $\cdots\langle i|M_{n-1}^{\dagger}M_{n-1}|0-2\rangle\langle 0-i|I_{n}|i\rangle=0$ for $d_{1}\leq i\leq d_{n-1}-1$. Thus we get $\langle i|M_{n-1}^{\dagger}M_{n-1}|0-2\rangle=0$, i.e., $m_{i,0}^{n-1}-m_{i,2}^{n-1}=0$. From Table~\ref{3_result}, $m_{i,2}^{n-1}=0$, which implies $m_{i,0}^{n-1}=m_{0,i}^{n-1}=0$. Therefore, all the off-diagonal entries of the matrix $E_{t}=(m_{i,j}^{t})_{i,j\in \mathbb Z_{d_{t}}}$ are zeros.
\begin{table*}
    \newcommand{\tabincell}[2]{\begin{tabular}{@{}#1@{}}#2\end{tabular}}
    \centering
    \caption{\label{3_result}Zero entries of the matrix $E_{t}=(m_{i,j}^{t})_{i,j\in \mathbb Z_{d_{t}}}$.}
    \begin{tabular}{cccc}
        \toprule
        \hline
        \hline
        \specialrule{0em}{1.5pt}{1.5pt}
        ~~~~~~~~~~~~~Sets~~~~~~~~~~~~~&~~~~~~~~~~~Pairs of states ~~~~~~~~~~~~~~~~&~~~~~~~~~~~~~~~~~Zero entries~~~~~~~~~~~~~~~~~&~~~~~~~Range ~~~~~~~\\
        \specialrule{0em}{1.5pt}{1.5pt}
        \midrule
        \hline
        \specialrule{0em}{1.5pt}{1.5pt}
        \tabincell{c}{$\mathcal {B}_{2}$\\$\mathcal {B}_{3}$}&\tabincell{c}{$|\phi_{i+(d_{1}-1)}\rangle$ \\$|\phi_{i+2(d_{1}-1)}\rangle$}&\tabincell{c}{$m_{i,0}^{1}= m_{0,i}^{1}= 0$}&\tabincell{c}{$1\leq i\leq d_{1}-1$} \\
        \specialrule{0em}{3.5pt}{3.5pt}
        \tabincell{c}{$\mathcal {B}_{2}$}&\tabincell{c}{$|\phi_{i+(d_{1}-1)}\rangle$ \\$|\phi_{j+(d_{1}-1)}\rangle$}&\tabincell{c}{$m_{i,j}^{1}= m_{j,i}^{1}= 0$}&\tabincell{c}{$1\leq i\neq j\leq d_{1}-1$} \\
        \specialrule{0em}{3.5pt}{3.5pt}
        \tabincell{c}{$\mathcal {B}_{t}$\\ $\mathcal {B}_{t+1}$}&\tabincell{c}{$|\phi_{i+(d_{1}-1)+\sum_{i=1}^{t-2}(d_{i}-1)}\rangle$ \\$|\phi_{i+(d_{1}-1)+\sum_{i=1}^{t-1}(d_{i}-1)}\rangle$}&\tabincell{c}{$m_{i,0}^{t-1}= m_{0,i}^{t-1}= 0$}&\tabincell{c}{$1\leq i\leq d_{t-1}-1$\\$3\leq t\leq n-1$} \\
        \specialrule{0em}{3.5pt}{3.5pt}
        \tabincell{c}{ $\mathcal {B}_{t}$}&\tabincell{c}{$|\phi_{i+(d_{1}-1)+\sum_{i=1}^{t-2}(d_{i}-1)}\rangle$ \\$|\phi_{j+(d_{1}-1)+\sum_{j=1}^{t-2}(d_{j}-1)}\rangle$}&\tabincell{c}{$m_{i,j}^{t-1}= m_{j,i}^{t-1}= 0$}&\tabincell{c}{$1\leq i\neq j\leq d_{t-1}-1$\\$3\leq t\leq n-1$} \\
        \specialrule{0em}{3.5pt}{3.5pt}
        \tabincell{c}{$\mathcal {B}_{1}$\\$\mathcal {B}_{n}$}&\tabincell{c}{$|\phi_{i}\rangle$ \\$|\phi_{i+(d_{1}-1)+\sum_{i=1}^{n-2}(d_{i}-1)}\rangle$}&\tabincell{c}{$m_{i,0}^{n-1}= m_{0,i}^{n-1}= 0$}&\tabincell{c}{$1\leq i\leq d_{1}-1$} \\
        \specialrule{0em}{3.5pt}{3.5pt}
        \tabincell{c}{$\mathcal {B}_{n}$}&\tabincell{c}{$|\phi_{i+(d_{1}-1)+\sum_{i=1}^{n-2}(d_{i}-1)}\rangle$ \\$|\phi_{j+(d_{1}-1)+\sum_{j=1}^{n-2}(d_{j}-1)}\rangle$}&\tabincell{c}{$m_{i,j}^{n-1}= m_{j,i}^{n-1}= 0$}&\tabincell{c}{$1\leq i \neq j\leq d_{n-1}-1$} \\
        \specialrule{0em}{3.5pt}{3.5pt}
        \tabincell{c}{$\mathcal {B}_{1}$\\$\mathcal {B}_{2}$}&\tabincell{c}{$|\phi_{i}\rangle$ \\$|\phi_{i+(d_{1}-1)}\rangle$}&\tabincell{c}{$m_{i,0}^{n}= m_{0,i}^{n}= 0$}&\tabincell{c}{$1\leq i\leq d_{1}-1$} \\
        \specialrule{0em}{3.5pt}{3.5pt}
        \tabincell{c}{$\mathcal {B}_{2}$\\$\mathcal {B}_{2n}$}&\tabincell{c}{$|\phi_{d_{1}+1}\rangle$ \\$|\phi_{i+\sum_{i=2}^{n}(d_{i}-1)}\rangle$}&\tabincell{c}{$m_{i,0}^{n}= m_{0,i}^{n}= 0$}&\tabincell{c}{$d_{1}\leq i\leq d_{n}-1$} \\
        \specialrule{0em}{3.5pt}{3.5pt}
        \tabincell{c}{$\mathcal {B}_{1}$}&\tabincell{c}{$|\phi_{i}\rangle$ \\$|\phi_{j}\rangle$}&\tabincell{c}{$m_{i,j}^{n}= m_{j,i}^{n}= 0$}&\tabincell{c}{$1\leq i\neq j\leq d_{1}-1$} \\
        \specialrule{0em}{3.5pt}{3.5pt}
        \tabincell{c}{$\mathcal {B}_{1}$\\$\mathcal {B}_{2n}$}&\tabincell{c}{$|\phi_{i}\rangle$ \\$|\phi_{j+\sum_{j=2}^{n}(d_{j}-1)}\rangle$}&\tabincell{c}{$m_{i,j}^{n}= m_{j,i}^{n}= 0$}&\tabincell{c}{$1\leq i \leq d_{1}-1$\\$d_{1}\leq j \leq d_{n}-1$} \\
        \specialrule{0em}{3.5pt}{3.5pt}
        \tabincell{c}{$\mathcal {B}_{2n}$}&\tabincell{c}{$|\phi_{i+\sum_{i=2}^{n}(d_{i}-1)}\rangle$ \\$|\phi_{j+\sum_{j=2}^{n}(d_{j}-1)}\rangle$}&\tabincell{c}{$m_{i,j}^{n}= m_{j,i}^{n}= 0$}&\tabincell{c}{$d_{1}\leq i\neq j\leq d_{n}-1$} \\
        \bottomrule
        \specialrule{0em}{1.5pt}{1.5pt}
        \hline
        \hline
    \end{tabular}
\end{table*}

By Lemma~\ref{lem:Dia}, we have $m_{0,0}^{t}= m_{i,i}^{t}$ with $1\leq t\leq n$ and $1\leq i\leq d_{t}-1$. It shows that all diagonal entries are equal from Table \ref{4_result}.
\begin{table}[H]
    \newcommand{\tabincell}[2]{\begin{tabular}{@{}#1@{}}#2\end{tabular}}
    \centering
    \caption{\label{4_result}Diagonal entries of $E_{t}=(m_{i,j}^{t})_{i,j\in \mathbb Z_{d_{t}}}$.}
    \begin{tabular}{ccc}
        \toprule
        \hline
        \hline
        \specialrule{0em}{1.5pt}{1.5pt}
        Pair of states~&Diagonal entries&~Range\\
        \specialrule{0em}{1.5pt}{1.5pt}
        \midrule
        \hline
        \specialrule{0em}{1.5pt}{1.5pt}
        \tabincell{c}{$|\phi_{i}\rangle$ \\$|\phi_{\sum_{i=2}^{n-1}d_{i}+2d_{n}-n+1}\rangle$}&\tabincell{c}{$m_{0,0}^{1}= m_{i,i}^{1}$}&\tabincell{c}{$1\leq i\leq d_{1}-1$} \\
        \specialrule{0em}{3.5pt}{3.5pt}
        \tabincell{c}{$|\phi_{i+(d_{1}-1)}\rangle$ \\$|\phi_{\sum_{i=2}^{n-1}d_{i}+2d_{n}-n+1}\rangle$}&\tabincell{c}{$m_{0,0}^{2}= m_{i,i}^{2}$}&\tabincell{c}{$1\leq i\leq d_{1}-1$} \\
        \specialrule{0em}{3.5pt}{3.5pt}
        \tabincell{c}{$|\phi_{i+(d_{1}-1)+\sum_{i=1}^{t-2}(d_{i}-1)}\rangle$ \\$|\phi_{\sum_{i=2}^{n-1}d_{i}+2d_{n}-n+1}\rangle$}&\tabincell{c}{$m_{0,0}^{t}= m_{i,i}^{t}$}&\tabincell{c}{$1\leq i \leq d_{t-1}-1$\\$3\leq t\leq n$} \\
        \specialrule{0em}{3.5pt}{3.5pt}
        \tabincell{c}{$|\phi_{i+\sum_{i=1}^{n-1}(d_{i}-1)}\rangle$ \\$|\phi_{\sum_{i=2}^{n-1}d_{i}+2d_{n}-n+1}\rangle$}&\tabincell{c}{$m_{0,0}^{t+1}= m_{i,i}^{t+1}$}&\tabincell{c}{$d_{t}\leq i \leq d_{t+1}-1$\\$1\leq t\leq n-2$}\\
        \specialrule{0em}{3.5pt}{3.5pt}
        \tabincell{c}{$|\phi_{i+\sum_{i=1}^{n-1}(d_{i}-1)}\rangle$ \\$|\phi_{\sum_{i=2}^{n-1}d_{i}+2d_{n}-n+1}\rangle$}&\tabincell{c}{$m_{(i-1),(i-1)}^{n}= m_{i,i}^{n}$}&\tabincell{c}{$d_{n-1}\leq i \leq d_{n}-1$}\\
        \bottomrule
        \specialrule{0em}{1.5pt}{1.5pt}
        \hline
        \hline
    \end{tabular}
\end{table}

In conclusion, all individual parties cannot start with a nontrivial measurement. The set of $\sum\limits_{i=2}^{n-1}d_{i}+2d_{n}-n+1$ orthogonal product states is indistinguishable by LOCC.
\end{proof}

In 2020, Jiang {\it et al}. presented $\sum\limits_{i=1}^{n}(2d_{i}-3)+1$ nonlocal quantum states in $\mathbb{C}^{d_{1}}\otimes \mathbb{C}^{d_{2}}\otimes\cdots\otimes \mathbb{C}^{d_{n}}$ \cite{Jiang2020}. The set of  nonlocal orthogonal product states proposed by us has less members in general.

\section{Conclusion}

The phenomenon of nonlocality without entanglement has attracted much attention in recent years. It is challenging to obtain an optimal set of nonlocal multipartite quantum states. In this paper, we have presented the set of $n(d-1)+1$ orthogonal product states which are indistinguishable by LOCC in $(\mathbb{C}^{d})^{\otimes n}$. Furthermore, we put forward the set of $\sum\limits_{i=2}^{n-1}d_{i}+2d_{n}-n+1$ nonlocal orthogonal product states for arbitrary multipartite quantum systems with less members. Note that the recent minimum size of $2n-1$ members in $\mathbb{C}^{m}\otimes \mathbb{C}^{n}$, proposed by Zhang {\it et al}. in Ref \cite{Zhangzc2016}, can be regarded as a special case of our results for multipartite quantum systems. It shows that our construction is universal and optimal.

\section*{Acknowledgments}

This work is supported by the Natural Science Foundation of Hebei Province (No. F2021205001); NSFC (Grants No. 62272208, No. 11871019, No. 12075159, No. 12171044); Beijing Natural Science Foundation (No. Z190005); the Academy for Multidisciplinary Studies, Capital Normal University; the Academician Innovation Platform of Hainan Province; and Shenzhen Institute for Quantum Science and Engineering, Southern University of Science and Technology (No. SIQSE202001).

\end{document}